\documentclass[a4paper]{article}
\usepackage[english]{babel}
\usepackage[T1]{fontenc}
\usepackage{amsmath}
\usepackage{amsfonts}
\usepackage{ulem}
\usepackage{amssymb}
\usepackage{amsthm}
\usepackage{epsfig}
\usepackage{graphicx}
\usepackage{array}
\usepackage{pxfonts}
\usepackage{setspace}
\usepackage{authblk}
\usepackage{amsmath}
\usepackage{amssymb}
\usepackage{nicefrac}
\usepackage{tikz}
\usepackage{pgfplots}
\usepgfplotslibrary{dateplot}
\usepackage{multido}
\usepackage{calc}
\newtheorem{theorem}{Theorem}
\newtheorem{proposition}{Proposition}

\def\is{\textsc{max independent set}}
\def\ds{\textsc{min dominating set}}
\def\np{\textbf{NP}}

\def\bb{{branch-and-bound}}
\def\gnp{$\mathcal{G}(n,p)$}

\let\geq\geqslant

\let\leq\leqslant

\makeatletter
\def\@fnsymbol#1{\ensuremath{\ifcase#1\or (a)\or (b)\or (c)\or (d)\or *\or \S \or
   \mathsection\or \mathparagraph\or \|\or **\or \dagger \or \ddagger \or \dagger\dagger
   \or \ddagger\ddagger \else\@ctrerr\fi}}
\makeatother

\usepackage{pgf,tikz}
\pgfplotsset{compat=1.14}
\title{Average-case complexity of a branch-and-bound algorithm for \ds{}}
\author{Tom Denat, Ararat Harutyunyan, Vangelis Th. Paschos \\
Universit\'e Paris-Dauphine, PSL* Research University, CNRS UMR 7243, LAMSADE, Paris 75016, France \\
\texttt{\{tom.denat,ararat.harutyunyan,vangelis.paschos\}@lamsade.dauphine.fr}}

\begin{document}

\maketitle

\begin{abstract}
The average-case complexity of a branch-and-bound algorithms for \ds{} problem in random graphs in the~$\mathcal{G}(n,p)$ model is studied.  We identify phase transitions between subexponential and exponential average-case  complexities, depending on the growth of the probability~$p$ with respect to the number $n$ of nodes.  
\end{abstract}

\section{Introduction}\label{intro}

Given a graph $G=(V,E)$ of order~$n$, a dominating set $S \subseteq V$ is a subset of~$V$ such that any $v_i \in V$ is either included in~$S$ or connected to a vertex of~$S$ by an edge of~$E$. The \ds{} problem consists of finding a minimum-size dominating in~$G$. \ds{} is a very well-known \np-hard problem completely equivalent (from both complexity and polynomial approximation points of view) to \textsc{min set cover problem}. 

Dealing with the exact solution of \ds{}, besides the obvious~$O(2^n)$ algorithm which considers the power set of~$V$ and chooses the smallest one that also forms a dominating set, several moderately exponential algorithms have been proposed mainly during the last fifteen years. 
To the best of our knowledge, the fastest one is the~$O(1.4969^n)$ algorithm due to~\cite{VANROOIJ20112147}.

The main purpose of this paper is the study of the average case complexity of branch-and-bound algorithms for the \ds{}  problem in random graphs in the $\mathcal{G}(n,p)$ model. This model represents graphs on~$n$ 
vertices where each of the possible $\binom{n}{2}$ edges appears independently with probability~$p$.
For an extensive treatment of random graphs, we refer the reader to the monograph~\cite{bo}.

Even though \bb{} is one of the best known and most widely used techniques for exactly solving \np{}-hard problems, there has been little systematic study of its complexity, worst- or average case.  Also, even though mathematical tools for average case analysis of algorithms have existed for decades~\cite{knuth1} and have much advanced in sophistication~\cite{flajolet-sedgewickbook}, we do not know of many results on the average case complexity of exact algorithms for \np{}-hard problems. The only works known to us are the ones of~\cite{banderieretal_siam} where the authors study the complexity of a ``pruning the search-tree algorithm" for \is{} (the worst-case complexity of this algorithm is~$O(1.3803^n)$,~\cite{Woeg2}) under the~$\mathcal{G}(n,p)$ model, the one of~\cite{BanderierHwangRavelomananaZacharovas2009}, where the same algorithm is studied under the~$\mathcal{G}(n,m)$ model and, finally, the one in~\cite{DBLP:journals/corr/BourgeoisCDP15} where the average-case complexity of a \bb{} algorithm for \is{} is studied  under the~$\mathcal{G}(n,p)$ model.

In what follows, in Section~\ref{bbalg} we specify the \bb{} algorithm the complexity of which is then analysed in Section~\ref{b&b}. Finally, in Section~\ref{simpleb&b} we study the complexity of simple exhaustive search algorithm which, starting from the whole vertex-set~$V$ of the input graph, produces a minimum dominating set by considering all the subsets of~$V$ and finally returns the smallest one that is a dominating set.


\section{The \bb{} algorithm}\label{bbalg}

Let $G = (V, E)$ be a graph; set $n=|V|$ and fix an order $v_1,v_2,\ldots,v_n$ on~$V$. The type of \bb{} algorithms for \ds{} studied here works by building a \bb{} binary tree, nodes of which are associated with a vector $\vec{x}\in \{0,1\}^{n}$ and a depth $\delta$ in the binary tree. Obviously, $x_i=1$ means that vertex~$v_i$ has been taken in the solution under construction and $x_i = 0$ means that~$v_i$ has not been taken. For a tree-node at level~$\delta$ only vertices $v_1, v_2, \ldots, v_{\delta}$ have been explored, i.e., only $x_1,x_2,\ldots,x_{\delta}$ 
have been assigned definite values. 
At this point the values for~$x_{\delta+1}, \ldots, x_n$ are, 
for the moment, equal to~1. We remark that the superset of a dominating set is also a dominating set. 

The root of the \bb{} tree~$T$ corresponds to the trivial dominating set including all the 
vertices ($\vec{x}=(1,1,...,1)$) at the depth 0 and is initially visited. 
The left child $(\vec{x},\delta)_l = (\vec{x},\delta+1) $ of a node $n_i=(\vec{x},\delta)$ at level
$\delta$ of~$T$, has exactly the same vector~$\vec{x}$ as~$n_i$ but with
the value of $x_{\delta+1}=1$ now determined, i.e., the partial solution represented by~$n_i$ 
is extended by putting~$v_{\delta+1}$ in the solution under construction. 
The right child~$(\vec{x},\delta)_r$ of~$(\vec{x},\delta)$ corresponds to 
changing the partial solution represented by~$n_i$ by putting $x_{\delta+1} = 0$.  

At each new step of the algorithm a new node will be explored. In order to be explorable a node~$(\vec{x},\delta)$ must correspond to a dominating set in~$G$ and be either the left or the right child of an already visited node. 

Therefore, at each step of the algorithm, the nodes can be divided in four categories:
\begin{enumerate}
\item the already visited nodes which correspond to dominating sets (feasible solutions);
\item the nodes that correspond to vertex-sets that are not dominating sets in~$G$ (the infeasible solutions);
\item the explorable solutions which correspond to dominating sets in~$G$ and are either the left or the right child of an already visited node;
\item  the currently ``hidden'' feasible solutions which correspond to dominating sets in~$G$ but are not the left nor the right child of an already visited node. 
\end{enumerate} 
For a node at level $\delta$ with vector $\vec{x}$, we define its score $u(\vec{x}, \delta)$ as
the number of vertices that currently must be included in the solution. Formally:
$$
u\left(\vec{x},\delta\right) =\left|\vec{x}\right| - n + \delta
$$ 
where $\vert \vec{x}\vert = \sum_{i=1,2, \ldots n} x_i$. Score~$u(\vec{x},\delta)$ can 
be seen as an optimistic prediction of the value of the optimal solution since it 
implies that no other vertex will be added to the feasible solution corresponding to the node at hand.

The choice of the next node to be explored among the explorable nodes (solutions) 
is made by minimizing a score~$u(\vec{x},\delta)$, named potential in what follows, 
associated with each node of the \bb{} tree.


The explorable solutions corresponding to tree-nodes of depth~$n$ constitute feasible dominating sets for the whole~$G$. It is easy to see that the leftmost among them  corresponds to a minimum dominating set of~$G$.

\begin{figure}[h!]
\centering
\includegraphics[scale=0.4]{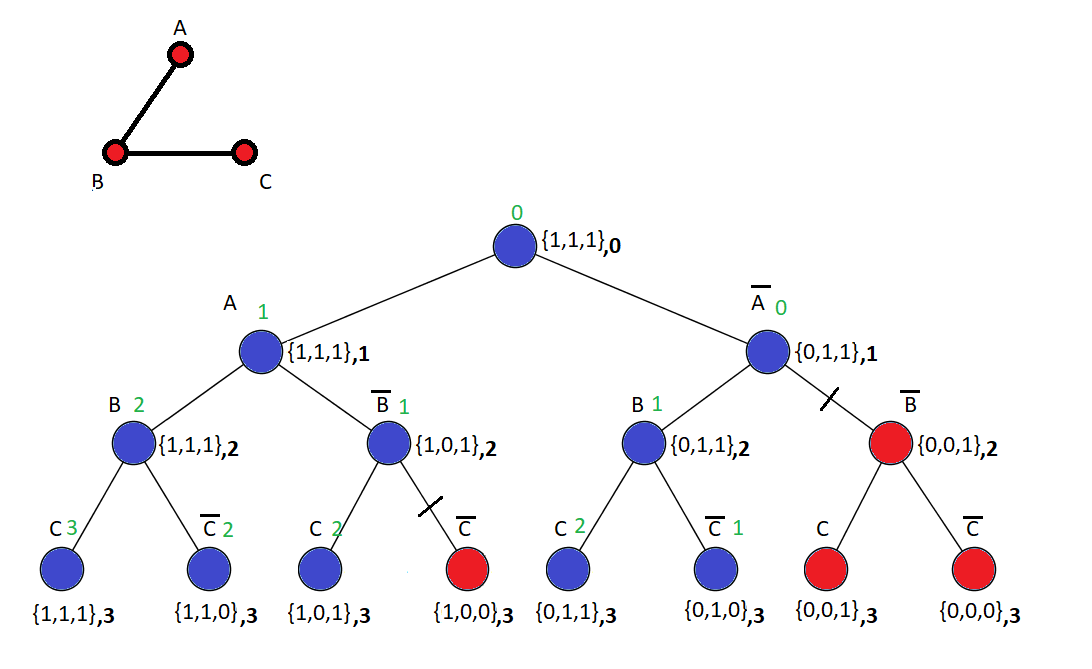}
\caption{Illustration of the \bb{} algorithm for the minimum dominating set problem. The instance~$G$ appears on the upper left corner of the figure. In the \bb{} tree, the nodes representing feasible solutions are coloured in deep blue. The green number above each node represents the score function~$u(\vec{x})$. 
}
\label{Necess}
\end{figure}

\begin{figure}[h!]
\centering
\includegraphics[scale=0.4]{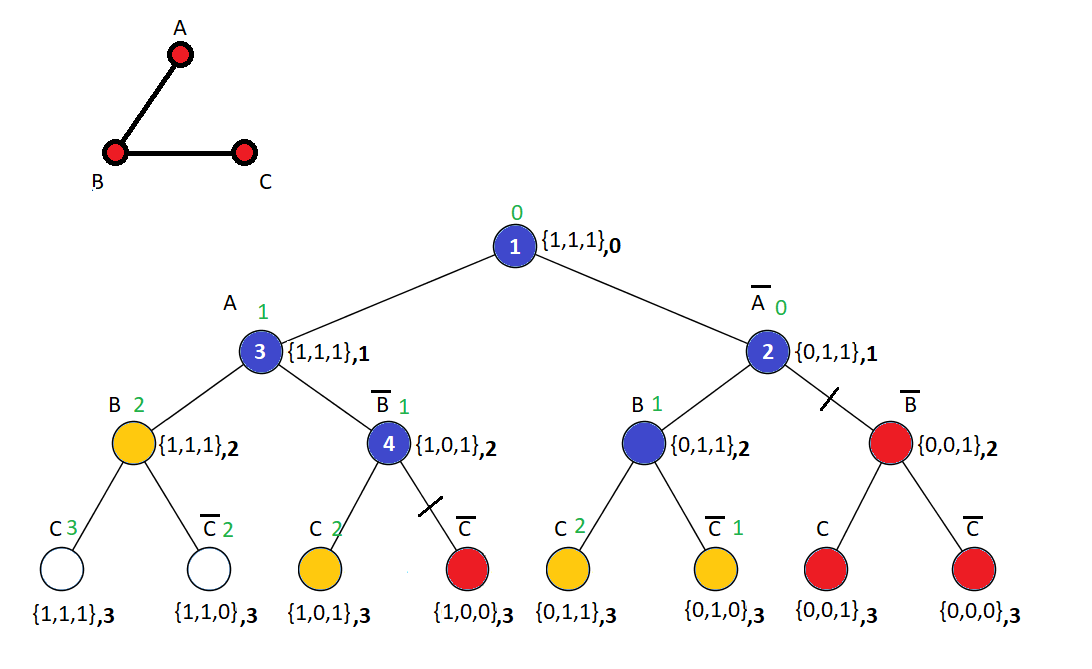}
\caption{Illustration of the fifth step of the algorithm. Here, the explored node are in deep blue, the explorable solutions are in yellow, the hidden solutions in white and the infeasible solutions are in red.}
\label{Necess1}
\end{figure}

\begin{figure}[h!]
\centering
\includegraphics[scale=0.4]{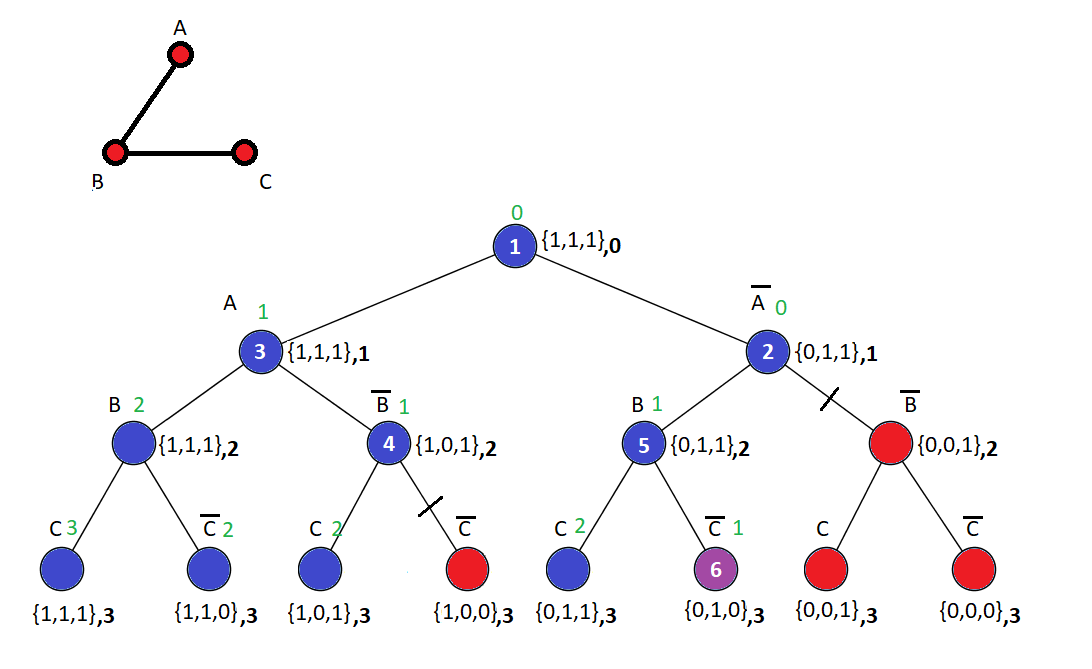}
\caption{Illustration of the whole execution of \bb{} algorithm for the minimum dominating set problem. The optimal solution is coloured in purple.}
\label{Necess2}
\end{figure}

As an example, on figure \ref{Necess} the graph~$G$ contains three vertices~$A$, $B$ and~$C$. Therefore, the \bb{} tree contains~3 levels. The branch and bound algorithm starts by exploring the root $\lbrace 1,1,1\rbrace ,0$. Now, the explorable nodes are $\lbrace 1,1,1\rbrace ,1$ and $\lbrace 0,1,1\rbrace ,1$ with respective potentials~$1$ and~$0$. Therefore, the node $\lbrace 0,1,1\rbrace ,1$ is explored next. Then, the two explorable nodes are $\lbrace 1,1,1\rbrace ,1$ and $\lbrace 0,1,1\rbrace ,2$, both with potential~$1$. The next node to be explored is chosen randomly. Assume that the next explored node is $\lbrace 1,1,1\rbrace ,1$ and so on. Finally, the complete solution $\lbrace 0,1,0\rbrace ,3$ is found. As its depth is 3, this solution $\lbrace B\rbrace$ is the minimum dominating set.

In what follows, we denote by~$\mathbb{T}(n,p)$ the average complexity of \bb{} in a binomial random graph $G=(V,E)$ with parameters~$(n,p)$, that is, the total number of leaves expansions, and by~$\mathbb{E}[\# S(G)]$ the expectation of the number of dominating sets in a binomial random graph~$G$.

\section{Analysis of the branch-and-bound algorithm}\label{b&b}

Notice first that any child~$xb$ of a node $x$ of the branch and bound tree, where $b\in \{0,1\}$,  has $u(xb) \geq u(x)$.

So, let~$x$ be the first leaf at the $n$th level --- that is with $|x|=n$ --- expanded by the algorithm; obviously $u(x) = |S_x|$. Any leaf that has not been expanded yet must have a bound that is at least~$|S_{x}|$ (otherwise they would have been handled before). Since all other unexpanded leaves have bounds no smaller than $|S_x|$, $S_x$ must be a minimum dominating set, and the algorithm terminates.  

It follows that, during the running of the algorithm, 
a leaf~$x$ is expanded only if $u(x) \leq |S^*|$, where $S^*$ is the minimum dominating set.  
Thus, the expanded leaves $x$ with $|x|=i$ are dominating subsets 
$S\subseteq \{1,\ldots,i\}$ satisfying $n-|S|\leq|S^*|$.  We can therefore write:
\begin{eqnarray}\label{tnp1}
\mathbb{T} &\leq& \sum\limits_{i=0}^{n}\sum\limits_{S\subseteq\{1,2,\ldots,i\}}\mathbb{P}\left(( \text{$S$ is a dominating set})\cap \left(n-|S| \leq \left|S^*\right|\right)\right)  \Rightarrow \nonumber \\
\Rightarrow
\mathbb{T} &\leq& \sum\limits_{i=0}^{n}\sum\limits_{S\subseteq\{1,2,\ldots,i\}}\mathbb{P}\left( \text{$S$ is a dominating set}\right) \;\; 
\wedge \;\;  \mathbb{T} 
\leq \sum\limits_{i=0}^{n}\sum\limits_{S\subseteq\{1,2,\ldots,i\}}\mathbb{P} \left(n-|S| 
\leq \left|S^*\right|\right) \nonumber \\ 
&\leq& n \cdot \max_{1 \leq k \leq n} \binom{n}{k} \Pr[\gamma > n -k]\\
\end{eqnarray}
Thus, we need to upper-bound the quantity $M:=\binom{n}{k} \Pr[\gamma > n -k]$ for all $ 1 \leq k \leq n$. 

\subsection{Upper bounds}
The following theorem provides upper bounds for the complexity of the \bb{} algorithm presented in Section~\ref{bbalg} for random graphs in the $\mathcal{G}(n,p)$ model. 
\begin{theorem}\label{positive} The following two facts hold:
\begin{enumerate}
\item[(a)] If $pn \to \infty$, then the \bb{} algorithm takes subexponential time.
\item[(b)] If $pn = c$, where $c \geq 20$ is a constant, then the \bb{} algorithm takes time $(2-\epsilon)^n$, where $\epsilon \geq 0.01$ is some constant.
\end{enumerate}
\end{theorem}
\begin{proof}
We first focus on the case where $p$ is fixed. We show that in this case,~$M$ is subexponential. 
We use the fact that $\gamma(G) \leq \alpha(G)$ for every graph~$G$, where~$\alpha(G)$ denotes the stability number of graph~$G$ (indeed, a maximal independent set is a dominating set) and the union bound. 

We recall that $1-x \leq e^{-x}$ for all values of~$x$ ( this follows from the fact that $1-x$ is the tangent line of $e^{-x}$ at $x=0$). We remark that $\alpha(G) > x$
for some~$x$, implies that one of the~$\binom{n}{x}$ subsets of vertices of size~$x$ induces an independent set in~$G$. Thus:
$$
M \leq \binom{n}{k} \Pr[\alpha > n - k] \leq \binom{n}{k} \binom{n}{n-k} (1-p)^{\binom{n-k}{2}} \leq \left(\binom{n}{n-k}\right)^2 e^{\nicefrac{-p(n-k)(n-k-1)}{2}} =
\binom{n}{x}^2 e^{-\nicefrac{px(x-1)}{2}}
$$
where $x = n-k$.
We consider~$x$ as a function of~$n$, setting $x:=f(n)$. If $f(n) = o(n)$, then $\binom{n}{x}^2 < (\nicefrac{en}{x})^{2x}$ is clearly subexponential. Thus, we may assume $f(n) = \Theta(n)$.
Quantity~$M$ clearly satisfies:
$$
M  \leq \left( \left(\frac{en}{x}\right)^2e^{\nicefrac{-p(x-1)}{2}} \right) ^x
$$
Since $f(n)= \Theta(n)$ and~$p$ is fixed, clearly, $(\nicefrac{en}{x})^2e^{\nicefrac{-p(x-1)}{2}} < 1$ for~$n$ sufficiently large. 
In fact, with a slightly more careful analysis, we can obtain that~$M$ is subexponential in the regime $pn \to \infty$. Indeed, we may suppose as before that $x = \Theta(n)$ and now it follows that $(\nicefrac{en}{x})^2$ is bounded by a constant and since $p = \omega(1)$
we have that $e^{-\nicefrac{p(x-1)}{2}} \to 0$, as $n \to \infty$. Therefore, $(\nicefrac{en}{x})^2e^{-\nicefrac{p(x-1)}{2}} < 1$ for~$n$ sufficiently large and the claim follows.

It remains to consider the case~(b), where $p = \nicefrac{c}{n}$, where $c \geq 20 $ is some absolute constant.  As before, we consider the quantity~$\binom{n}{x}^2 e^{\nicefrac{-px(x-1)}{2}}$. We may assume as before that $x = \Omega(n)$, i.e. $\lim_{x \to \infty} \nicefrac{x}{n} = \epsilon$, for some $\epsilon > 0$. Indeed,
otherwise the time is subexponential.
Thus, we need to show that $\binom{n}{\epsilon n}^2 (1-p)^{\nicefrac{\epsilon^2 n^2}{2}}$ is bounded for any $\epsilon > 0$. We use the fact that $\binom{n}{\epsilon n} = 2^{(1+o(1)) H(\epsilon) n} $, where
$H(\epsilon)$ is the binary entropy function. Thus, 
\begin{equation*}
\binom{n}{\epsilon n}^2 (1-p)^{\nicefrac{\epsilon^2 n^2}{2}} \leq 
\left(\left[\left(\epsilon^{-\epsilon}(1-\epsilon)^{-(1-\epsilon)}\right)^{(1+o(1))}\right]^2 
e^{-c \nicefrac{\epsilon^2}{2}} \right)^n
\end{equation*}
We remark that $(\epsilon^{-\epsilon}(1-\epsilon)^{-(1-\epsilon)})$ is increasing on the interval $(0, \nicefrac{1}{2})$ 
and decreasing on $(\nicefrac{1}{2}, 1)$, with its maximum at $\epsilon = \nicefrac{1}{2}$, since $H(\epsilon)$ 
is the logarithm of this function. Thus, there is some constant $\epsilon_0 < \nicefrac{1}{2}$ such that, 
for all $\epsilon \leq \epsilon_0$, $(\epsilon^{-\epsilon}(1-\epsilon)^{-(1-\epsilon)}) < \sqrt{2}$. It is easy to check 
that we can take $\epsilon_0 = \nicefrac{1}{10}$. Thus, we may assume that $\epsilon > \nicefrac{1}{10}$.
Since $c\geq 20$, it suffices to show that:
$$
\left[\left(\epsilon^{-\epsilon}(1-\epsilon)^{-(1-\epsilon)}\right)^{(1+o(1))}\right]^2 e^{-10 \epsilon^2} < 2
$$
for all $\epsilon \in [\nicefrac{1}{10}, \nicefrac{1}{2}]$. Since the first of the two products is 
increasing and the second is decreasing with~$\epsilon$, we will have to dominate each separately. 
To this end, we refine the intervals~of $\epsilon$. First consider the 
interval $\epsilon \in [\nicefrac{1}{10}, \nicefrac{1}{8}]$. To bound our product, 
it is sufficient to substitute~$\nicefrac{1}{8}$ in the first term and~$\nicefrac{1}{10}$ in the second. 
By doing this, we obtain that the product is less than, say,~1.99. It is easily verified 
that we can repeat  
this argument on the following intervals, thus finishing the theorem. 
In each case, we obtain a bound of less than~1.99. The precise bounds are
given below, where $f(\epsilon) := [(\epsilon^{-\epsilon}(1-\epsilon)^{-(1-\epsilon)})^{(1+o(1))}]^2 e^{-10 \epsilon^2}$:
\begin{itemize}
 \item $\epsilon \in [1/8, 1/7]; f(\epsilon) < 1.943$;
 \item $\epsilon \in [1/7, 0.15]; f(\epsilon) < 1.9$;
 \item $\epsilon \in [0.15, 0.17]; f(\epsilon) < 1.988$;
 \item $\epsilon \in [0.17, 0.19]; f(\epsilon) < 1.981$;
 \item $\epsilon \in [0.19, 0.21]; f(\epsilon) < 1.95$;
 \item $\epsilon \in [0.21, 0.25]; f(\epsilon) < 1.982$;
 \item $\epsilon \in [0.25, 0.35]; f(\epsilon) < 1.955$;
 \item $\epsilon \in [0.35, 0.5]; f(\epsilon) < 1.2$.
\end{itemize}
%
%
%
The proof of the theorem is now completed.~\end{proof}

\subsection{Lower bounds}
The following result shows that the upper bound on complexity of the algorithm given by Item~(b) of Theorem~\ref{positive} cannot be  drastically improved in order that a subexponential bound is taken.
\begin{theorem} 
Let $p = \nicefrac{c}{n}$, where~$c$ is a positive fixed constant. Then, the branch-and-bound algorithm takes at least $(\nicefrac{1}{\epsilon})^{\epsilon n}$
time for~$G(n,p)$, where $\epsilon:= \max\{0.99, 1-(\nicefrac{1}{10c})\}$ 
\end{theorem}
%
\begin{proof}
It is sufficient to show that:
%
$$
\binom{n}{\epsilon n} \Pr[\gamma > n - \epsilon n] = \Omega\left(\left(\frac{1}{\epsilon}\right)^{\epsilon n}\right)
$$ 
for~$n$ sufficiently large.

We will prove that  $\Pr[\gamma \leq n - \epsilon n]$ is arbitrarily small. This is clearly sufficient. Let~$A$ be the event that a fixed set~$S$ of size $n - \epsilon n$ is a dominating set. Then, $\Pr[\gamma \leq n - \epsilon n] \leq \binom{n}{n-\epsilon n}\Pr[A] < 2^n \Pr[A]$.
We use the fact that $1-x \geq e^{-2x}$ for all $x \in (0, \nicefrac{1}{2})$; this can be seen, for example,
by noticing that $e^{-2x}$ is a convex function and that $1-x = e^{-2x}$ has two solutions at $x=0$ and at some $x \in (0.5, 1)$.
Now, $\Pr[A] = \left(1- (1-p)^{n - \epsilon n} \right)^{\epsilon n} \leq (1- e^{-2c(1-\epsilon)})^{\epsilon n}$.
Thus:
$$ 
\Pr[\gamma \leq n - \epsilon n] \leq \left(2 \left(1- e^{-2c(1-\epsilon)}\right)^{\epsilon}\right)^n
$$
Now, if $c< 10$, then $(1- e^{-2c(1-\epsilon)})^{\epsilon} < (1 - e^{-\nicefrac{1}{5}})^{0.99} < 1/2$. Similarly, if $c \geq 10$, then $(1- e^{-2c(1-\epsilon)})^{\epsilon} < (1 - e^{-\nicefrac{1}{5}})^{1- \nicefrac{1}{10c}} < 1/2$. Thus, $\Pr[\gamma \leq n - \epsilon n]< \delta^n$, for some $\delta < 1$.

It follows that:
$$
\binom{n}{\epsilon n} \Pr[\gamma > n - \epsilon n] = \Omega\left(\left(\frac{1}{\epsilon}\right)^{\epsilon n}\right)
$$
as required.
\end{proof}

\section{Analysis of the simple exhaustive search algorithm}\label{simpleb&b}

We conclude the paper by studying in this section a simple exhaustive search algorithm which, starting from the whole vertex-set~$V$ of the input graph produces a minimum dominating set by considering all the subsets of~$V$ and finally returns the smallest one that is a dominating set. In the remaining part of this section, we prove the following proposition.
\begin{proposition} Consider  a random $(n,p)$-binomial graph~$G$. Then the following hold:
\begin{enumerate}
\item if~$p$ is smaller than~$\nicefrac{1}{n}$, then the complexity of the exhaustive search algorithm is subexponential;
\item if $p =  \nicefrac{c}{n}$, for some constant $c > 1$, then:
$$
\mathbb{T}(n,p) \leq \max\left\{1.99^n, \left(2\left(1 - e^{-2c}\right)^{1/3}\right)^{n}\right\}
$$
\end{enumerate}
\end{proposition}
\begin{proof}
We first prove item~1. Consider  a random $(n,p)$-binomial graph~$G$. Recall that by definition of the $(n,p)$-binomial random model, the probability that a set of which~ $k$ vertices are excluded is a dominating set of $G$ is equal to $(1-(1-p)^{n-k})^k$; henceforth:
\begin{equation}\label{es}
\mathbb{E}[\# S(G)] = \sum_{k=0}^n{n \choose k}\left(1-(1-p)^{n-k}\right)^k \leq n \cdot\max\limits_{k \leq n}\left\{{n \choose k}(1-(1-p)^{n-k})^k \right\}
\end{equation}
The number of dominating sets in the sub-graph induced by the~$\alpha$ first fixed vertices is strictly greater than the number of dominating sets induced in the $\alpha-1$ first such vertices (since the~$\alpha$-th vertex is a dominating set by itself); so, using~(\ref{es}):
\begin{equation}\label{boundcn}
\mathbb{T}(n,p) \leq n\cdot\mathbb{E}[\# S(G)] \leq n \sum_{k=0}^n{n \choose k}\left(1-(1-p)^{n-k}\right)^k \leq n^2 \cdot\max\limits_{k \leq n}\left\{{n \choose k}(1-(1-p)^{n-k})^k \right\}
\end{equation}
Obviously, the more the edges in the graph the more likely is that a sub-graph is a dominating set and therefore the probability that a sub-graph is a dominating set increases with~$p$ and decreases with $q=1-p$. Let us consider the case $p=\nicefrac{j}{n}$. We will thereafter consider that $k \sim i n$ with $0<i<1$. Indeed, when $k< i n, \forall i>0$ (resp. $k> i n, \forall i<1$) then ${n \choose k}$ is at most subexponential. Hence, the probability for a graph of size $n-k $ to be a dominating set is: 
$$
\left(1-\left(1-\frac{j}{n}\right)^{(1-i)n}\right)^{in}
$$ 
which when $n$ tends to infinity tends to $(1-e^{j(i-1)})^{in}$.

Using~(\ref{boundcn}), discussion just above leads to:
\begin{equation}\label{tnp}
\mathbb{T}(n,p) \leq n^2 \cdot\max\limits_{0\leq i \leq 1}\left\{{n \choose in}(1-e^{j(i-1)})^{in} \right\}\leq  n^2 \cdot\max\limits_{0\leq i \leq 1}\left\{\frac{e^{i n}}{i^{i n}}(1-e^{j(i-1)})^{in} \right\} 
\end{equation}
Let $f(x)=(\nicefrac{e^{x}}{x^{x}})(1-e^{j(x-1)})^{x}$. Then:
$$
f'(x)= -e^x \left(\frac{1}{x}\right)^x \left(1 - e^{j(x - 1)}\right)^{x - 1} \left(\left(e^{j(x - 1)} - 1\right) \left(\log\left(1 - e^{j(x - 1)}\right) - \log(x)\right)\right)
$$
which gives: $f'(x)=0 \Leftrightarrow 1 - e^{j(x - 1)}=x \Leftrightarrow x= 1- (\nicefrac{W(j)}{j})$. So:
$$
\max\limits_{0\leq x \leq 1}f(x) = f\left(1- \frac{W(j)}{j}\right)=e^{1- \frac{W(j)}{j}}
$$
where $W(\cdot)$ denotes the {Lambert}'s function defined by $x  =W(x)e^{W(x)}$.
Let $g_+(j)=e^{1- (\nicefrac{W(j)}{j})}$. Then, using~(\ref{tnp}), we get:  $\mathbb{T}(n,p)\leq n^2 \cdot g_+(j)^n$.

Note that $g_+(0)=1$ and $\lim_{x \to \infty} g_+(x)=e$. Then:
$$
\mathbb{T}(n,p) \geq \max\limits_{0\leq i \leq 1}\left\{{n \choose in}(1-e^{j(i-1)})^{in} \right\} \geq \max\limits_{0\leq i \leq 1}\left\{\frac{1}{(i)^{i n}}(1-e^{j(i-1)})^{in} \right\} \geq \exp\left(\frac{1}{e}-\left(\frac{W\left(j e^{-j-1+\frac{j}{e}}\right)}{j}\right) \right)^n = g_-^n(x)
$$
Observe that $\lim_{x \to \infty} g_-(x)=e^{\nicefrac{1}{e}}\simeq 1.44$.

Given that the complexity of the exhaustive search algorithm is increasing with $p$, discussion above concludes that it is subexponential if and only if~$p$ is smaller than~$\nicefrac{1}{n}$ and the proof of item~1 is now complete.

We now prove item~2. Here, we can clearly suppose $n/3 < k < 2n/3$; otherwise $\binom{n}{k} \leq 1.99^n$ and we are done. Using the fact that $1-x \geq e^{-2x}$ for $x \in (0, 0.5)$, one can deduce:
 \begin{eqnarray*}
  \mathbb{T}(n,p) &\leq& 2^n n^2 \max\limits_{\nicefrac{n}{3} \leq k \leq \nicefrac{2n}{3}}\left\{{n \choose k}\left(1-(1-p)^{n-k}\right)^k \right\} \\
  &\leq& 2^n n^2 \left(1 - \left(1-\frac{c}{n}\right)^{\nicefrac{2n}{3}}\right)^{\nicefrac{n}{3}} \;\; \leq \;\; 2^n n^2 \left(1 - e^{-\nicefrac{4c}{3}}\right)^{\nicefrac{n}{3}} \;\; \leq \;\; n^2 \left(2 \left(1 - e^{-\nicefrac{4c}{3}}\right)^{\nicefrac{1}{3}}\right)^{n}
 \end{eqnarray*}
and the result for item~2 follows immediately.
\end{proof}

\section{Conclusion}

We have studied in this paper the average-case complexity of a \bb{} algorithm  for \ds{} in 
random graphs under the \gnp{} model. It has been proved 
that this complexity is: (a)~\textit{subexpontial} when $p = \nicefrac{f(n)}{n}$, 
for any function $f \rightarrow \infty$ with~$n$; (b)~\textit{exponential} when $p= \nicefrac{c}{n}$. 
For the latter case it was proved that the smaller the constant~$c$ the closer to~$2^n$ 
average case complexity of the algorithm. 
Then the complexity of a naive exhaustive search algorithm has been studied. 
Here, for $p$ smaller than~$\nicefrac{1}{n}$ 
the algorithm is subexponential, while for $p \geqslant \nicefrac{c}{n}$, $c > 1$, its 
complexity becomes exponential, tending, for very large values of~$c$, to~$2^n$.

\bibliographystyle{plain}
\bibliography{biblio}

\end{document}